\documentclass[submission,copyright,creativecommons]{eptcs}

\usepackage{iftex}

\ifpdf
  \usepackage{underscore}         
  \usepackage[T1]{fontenc}        
\else
  \usepackage{breakurl}           
\fi

\RequirePackage{amsmath}
\RequirePackage{amssymb}
\RequirePackage{amsthm}

\usepackage{stmaryrd}

\usepackage[all,2cell,cmtip]{xy}
\UseTwocells

\newtheorem{theorem}{Theorem}[section]

\theoremstyle{definition}

\newtheorem{definition}{Definition}[section]
\newtheorem{warning}{Warning}[section]

\usepackage{import}
\usepackage{tikz}
\usepackage{amsmath,amsfonts,amssymb}
\usepackage{graphicx}

\newcommand\tombstone{\qed}
\newcommand{\from}{\leftarrow}

\newcommand{\Hom}{\mathrm{Hom}}

\newcommand{\tensor}{\otimes}
\newcommand{\Tr}{\mathrm{Tr}}

\newcommand{\Cat}{\mathcal{C}}
\newcommand{\Bicat}{\mathcal{B}}
\newcommand{\Mat}{\mathrm{Mat}}
\newcommand{\MCat}{\Mat(\Cat)}
\newcommand{\MatC}{\Mat(\Cat)}
\newcommand{\Item}{$(\ast)$\ }

\newcommand{\Natural}{\mathbb{N}}

\newcommand{\FVec}{\mathbf{FVec}}

\newcommand{\lifteq}{\raisebox{20pt}{$=$}\ }

\newcommand{\UShow}[1]{\raisebox{-0.4\height}{\includegraphics[scale=0.8]{mat_c/#1.pdf}}}
\newcommand{\Show}[1]{\raisebox{-0.4\height}{\includegraphics[scale=1.0]{mat_c/#1.pdf}}}

\newcommand{\Axiom}[1]{{\tt(#1)}}

\title{Categorified Path Calculus}

\author{Simon Burton
\institute{Quantinuum\\
Terrington House, 13-15 Hills Road, Cambridge CB2 1NL, United Kingdom}
\email{simon.burton@quantinuum.com}
}

\date{\today}
\begin{document}
\maketitle

\begin{abstract}
Path calculus, or 
graphical linear algebra, is a string diagram calculus for
the category of matrices over a base ring.
It is the usual string diagram calculus for a symmetric monoidal
category, where the monoidal product is the direct sum of matrices.
We categorify this story to 
develop a surface diagram calculus for
the bicategory of matrices over a base bimonoidal category.
This yields a surface diagram calculus for any bimonoidal category 
by restricting to diagrams for $1\times 1$ matrices.
We show how additional structure on the base category,
such as biproducts, duals and the dagger, adds structure to the 
resulting calculus.
Applied to categorical quantum mechanics this yields a
new graphical proof of the teleportation protocol.
\end{abstract}

\newcommand{\Comp}{\circ}



\section{Introduction: unifying the old and new quantum}


What is an amplitude?
Feynman's account of quantum mechanics
epitomizes the particle physicist's conception of amplitudes:
``counting'' the paths from a source to a detector (\cite{Feynman1965} \S 3.1).
The rules for such path counting are 
(i) parallel paths add, and (ii) serial paths multiply:
$$
2 + 3 =
\raisebox{-0.4\height}{\includegraphics{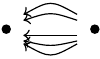}}
= 5,
\ \ \ \ 
2 \times 3 = 
\raisebox{-0.4\height}{\includegraphics{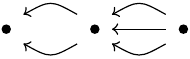}}
= 6.
$$
This is a cartesian, or classical, viewpoint on amplitudes,
and would correspond to a classical dynamical
system with uncertainty such as the Galton board,
if it wasn't for the Born rule. 
\emph{Path calculus} or 
\emph{Graphical linear algebra} neatly encapsulates this
correspondence between linear algebra and path counting
as a \emph{bimonoid in a monoidal category}~\cite{Sobocinski2013,deFelice2022}.


Around the same time Feynman was giving his lectures,
John Bell was initiating the study of non-locality in quantum
physics, otherwise known as entanglement~\cite{Bell1964}.
This also has a categorical interpretation as
adjointness (duals) in a compact closed category~\cite{Abramsky2004}.
Perhaps this is why particle physicists (before Bell) 
didn't notice entanglement: infinite dimensional
Hilbert spaces don't have adjoints (duals).

There is a fundamental incompatibility here between the
(bi-)cartesian monoidal structure of the \emph{old quantum},
and the multiplicative (tensor) monoidal structure of the \emph{new quantum}.
In this work we resolve this conflict by \emph{categorification}
or \emph{2-linear algebra}~\cite{Kapranov1994,Baez1997}.

We consider this work as
part of 2-categorical quantum mechanics initiated by 
Vicary~\cite{Vicary2012}.

In section \ref{userguide} we present a tour of the graphical
calculus as applied to categorified matrices.
This is meant to be as concrete as possible.
We show how additional structure on the base category,
such as biproducts \ref{biproducts}, duals \ref{duals} and the dagger \ref{teleportation},
is represented in the resulting calculus.
Applied to categorical quantum mechanics this yields a
new graphical proof of the teleportation protocol \ref{teleportation}.
We also show how this graphical calculus relates to 
similar calculi in the literature \ref{relation}.
Finally, section \ref{bicategory} is a recapitulation of
the main ideas, from a more abstract perspective.

\section{A user's guide to categorified linear algebra}\label{userguide}

\subsection{Review of path calculus}\label{path-calculus}

As a perceptual warmup, we review path calculus, or graphical linear algebra,
using conventions that fit with the development below.
At an elementary level, the story is about the coincidence
between path-counting and matrices with natural number entries.
A matrix with $m$ rows and $n$ columns corresponds to
a path diagram with $m$ outputs, and $n$ inputs.
We draw diagrams algebraically, from right to left.

\begin{center}
\begin{tabular}{|c|c|c|c|c|c|}
\hline
diagram & \Show{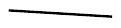} & \Show{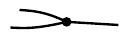} & \Show{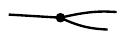} & \Show{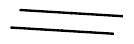} & \Show{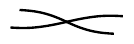} \\
\hline
matrix & 
$\begin{bmatrix}1\end{bmatrix}$ & 
$\begin{bmatrix}1\\1\end{bmatrix}$ & 
$\begin{bmatrix}1& 1\end{bmatrix}$ &
$\begin{bmatrix}1&0\\0&1\end{bmatrix}$ &
$\begin{bmatrix}0&1\\1&0\end{bmatrix}$ 
\\
\hline
\end{tabular}
\end{center}
This last diagram we call a \emph{swap}.

Horizontal composition of diagrams corresponds to matrix
multiplication:

\begin{center}
\begin{tabular}{c c c}
\Show{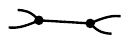} & & \\
\hspace{10pt}$\begin{bmatrix}1\\1\end{bmatrix}\begin{bmatrix}1& 1\end{bmatrix}$ & $=$ & $\begin{bmatrix}1&1\\1&1\end{bmatrix}$
\end{tabular}
\begin{tabular}{c}
\Show{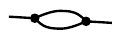} \\
\hspace{20pt}$ \begin{bmatrix}1& 1\end{bmatrix}\begin{bmatrix}1\\1\end{bmatrix} =\begin{bmatrix}2\end{bmatrix}$
\end{tabular}
\end{center}

We are drawing these diagrams skewed because we are saving the third dimension
for later. The coordinate system looks like this:
\begin{center}
\Show{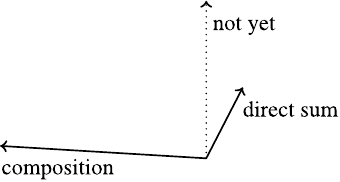}
\end{center}

Stacking diagrams into the page corresponds to direct sum of matrices:
\begin{center}
\begin{tabular}{r c r c}
 & \Show{nn} & & \Show{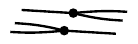} \\
$\begin{bmatrix}1\end{bmatrix}\oplus \begin{bmatrix}1\end{bmatrix}=$\hspace{-10pt}
 & $\begin{bmatrix}1&0\\0&1\end{bmatrix}$ &   
$\begin{bmatrix}1\\1\end{bmatrix}\oplus \begin{bmatrix}1&1\end{bmatrix}=$\hspace{-10pt}
&
$\begin{bmatrix}1&0&0\\1&0&0\\0&1&1\end{bmatrix}$ 
\end{tabular}
\end{center}

Empty matrices come in different sizes, which we notate using subscripts:
\begin{center}
\begin{tabular}{l r r}
\Show{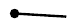} & \Show{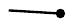}\ \ \ \  & \\
$\begin{bmatrix}\ \end{bmatrix}_{0,1}$ &
$\begin{bmatrix}\ \end{bmatrix}_{1,0}$ & 
\ \ \ $\begin{bmatrix}\ \end{bmatrix}_{0,0}$ 
\end{tabular}
\end{center}
The matrix with zero rows and zero columns has the empty diagram.

Using these building blocks, we have equations which we render using the vertical
dimension:
\begin{center}
\Axiom{l-unit}\ \UShow{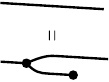} \hspace{20pt}
\Axiom{r-unit}\ \UShow{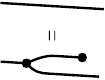} \hspace{20pt}
\Axiom{assoc}\ \UShow{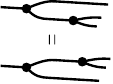}   \hspace{20pt}
\Axiom{comm}\ \UShow{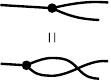}
\end{center}
and the horizontal opposites:
\begin{center}
\Axiom{l-counit}\hspace{0pt}\UShow{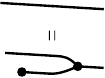} \hspace{20pt}
\Axiom{r-counit}\hspace{0pt}\UShow{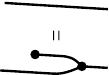} \hspace{20pt}
\Axiom{coassoc}\hspace{0pt}\UShow{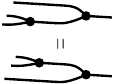} \hspace{20pt}
\Axiom{cocomm}\hspace{0pt}\UShow{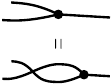}
\end{center}
Finally the equations:
\begin{center}
\Axiom{bimonoid}\UShow{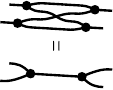}    \hspace{10pt}
\Axiom{comul-unit}\UShow{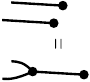} \hspace{10pt}
\Axiom{counit-mul}\UShow{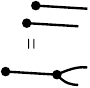} \hspace{10pt}
\Axiom{counit-unit}\UShow{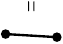} 
\end{center}
The punchline of graphical linear algebra is that using this equational
diagrammatic theory we can recover the theory of matrices with natural number entries.
(Strictly speaking, we also need isotopy equations for the swaps.)

We can go a bit further with this correspondence. 
Given any semi-ring, or rig, $(R,+,\times,0,1)$ we introduce the
diagrammatic generators:
$$
\Show{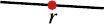}
$$
for each $r\in R$.
These satisfy equations
\begin{center}
\Axiom{unit-hom} \UShow{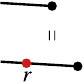}   \hspace{10pt}
\Axiom{mul-hom} \UShow{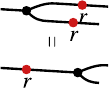}   \hspace{10pt}
\Axiom{counit-hom} \UShow{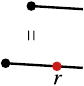}   \hspace{10pt}
\Axiom{comul-hom} \UShow{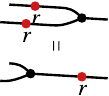}
\end{center}

\begin{center}
\Axiom{add} \UShow{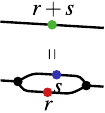}   \hspace{20pt}
\Axiom{zero} \UShow{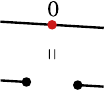}   \hspace{20pt}
\Axiom{mul} \UShow{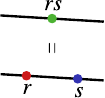}   \hspace{20pt}
\Axiom{one} \UShow{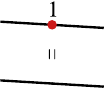}
\end{center}
as well as equations for interacting with the swaps.

As a first step in our categorification journey we notice that
the natural numbers, as a semi-ring or rig:
$$
    (\Natural, +, \times, 0, 1)
$$
can be replaced, or \emph{categorified}, by the category
of finite dimensional vector spaces over a field, $\FVec$.
This category has enough structure to mimic the
addition and multiplication of natural numbers:
$$
    (\FVec, \oplus, \otimes, O, I),
$$
where $O$ is a fixed zero dimensional vector space, and $I$ is a fixed
one dimensional vector space.
The equational axioms above are replaced by specific isomorphisms,
which then satisfy further equations.
This leads to the theory of Kapranov-Voevodksy 2-vector spaces.
However, in the following section we will proceed in more generality,
replacing $\FVec$ with an arbitrary \emph{bimonoidal category} $\Cat$.




\subsection{Categorified linear algebra}

Given a semi-ring or rig, $(R,+,\times,0,1)$
we can form the \emph{category of matrices} $\Mat(R)$
over this rig, whose objects are natural numbers,
morphisms $m\leftarrow n$ are $m\times n$ matrices with entries in $R$
and composition is given by matrix product.
We now wish to replace the rig $R$ with a category $\Cat$ that has
enough structure to perform rig-like operations on the objects of $\Cat$.

A \emph{bimonoidal category} is a category $\Cat$
equipped with an \emph{additive} symmetric monoidal structure,
$$(\Cat, \oplus, O, \alpha^\oplus, \lambda^\oplus, \rho^\oplus, \sigma^\oplus)$$
and a \emph{multiplicative} monoidal structure,
$$(\Cat, \tensor, I, \alpha^\tensor, \lambda^\tensor, \rho^\tensor)$$
such that $\tensor$ \emph{distributes} over $\oplus$
via natural isomorphisms:
\begin{align*}
    \delta^l_{A,B,C} &: A\tensor(B\oplus C) \to (A\tensor B)\oplus(A\tensor C) \\
    \delta^r_{A,B,C} &: (A\oplus B)\tensor C \to (A\tensor C)\oplus(B\tensor C) \\
    \lambda_A &: O\tensor A \to O \\
    \rho_A &: A\tensor O \to O 
\end{align*}
called respectively the left- and right-\emph{distributors}, and
the left- and right-\emph{nullitors}.
All these structures satisfy coherence equations~\cite{Laplaza1971,Kelly1974}.
We denote $\Cat_0$ as the set of objects of $\Cat$, and $\Cat_1$ as
the set of morphisms of $\Cat$.

We now consider the bicategory $\MatC$ of matrices over a bimonoidal category $\Cat$.
Important examples are, $\Cat=\{\mathtt{true}, \mathtt{false}\}$
the poset of truth values considered as a bimonoidal category
($\oplus$ and $\tensor$ are disjunction and conjunction), in which
case $\MatC$ is the bicategory of finite sets and relations,
and $\Cat=\mathrm{FdVec}_\mathbb{K}$ the bimonoidal category of finite dimensional vector spaces over
a field $\mathbb{K}$. This case gives $\MatC$ as the bicategory of
Kapranov-Voevodksy 2-vector spaces~\cite{Kapranov1994}.

In this section we dive into example calculations in $\MatC$, learning on-the-job.
See section \ref{bicategory} for a more formal treatment, and 
the reference~\cite{Johnson2021} chapter 8, for an exquisitely detailed
definition of $\MatC$.

As a bicategory, $\MatC$ has objects the natural numbers $\Natural=\{0,1,2...\}$, 
morphisms $m\leftarrow n$ are $m\times n$ matrices of objects of $\Cat$,
and 2-morphisms 
$ \xymatrix{ m & \ltwocell{} n }$
are $m\times n$ matrices of morphisms of $\Cat$.
We call these the \emph{0-cells}, \emph{1-cells} and \emph{2-cells}
respectively.
We render $n\times m$ matrices in $\MatC$ with $m$ incoming
surfaces and $n$ outgoing surfaces. Here are some 1-cells built
using the object $I\in\Cat_0$:
\begin{center}
\begin{tabular}{c c c}
\raisebox{-0.4\height}{\ \includegraphics{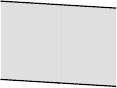}\ }
&
\raisebox{-0.4\height}{\ \includegraphics{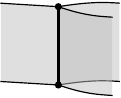}\ }
&
\raisebox{-0.4\height}{\ \includegraphics{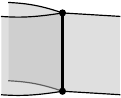}\ }
\\
$[I]$
&
$[I\ \ I]$
&
$\begin{bmatrix}I\\I\end{bmatrix}$
\end{tabular}
\end{center}
Where the anonymous grey surface is the 0-cell $1\in\Natural$.
These surface diagrams flow from right-to-left,
which matches the usual algebraic notation.
Horizontal composition is categorified matrix multiplication: 
instead of using  $+,\times$ of a rig $R$, we use 
the bimonoidal structure $\oplus,\tensor$ of $\Cat$.
Here we show the horizontal composition of some 1-cells:

\begin{center}
\begin{tabular}{l l}
\raisebox{-0.4\height}{\ \includegraphics{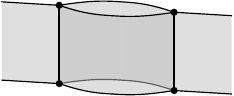}\ }
&
\raisebox{-0.4\height}{\ \includegraphics{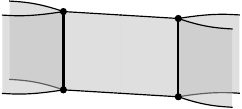}\ }
\\
$
[I\ \ I] \begin{bmatrix}I\\I\end{bmatrix} = [I\tensor I \oplus I\tensor I]
\cong [I\oplus I]
$
&
$
\begin{bmatrix}I\\I\end{bmatrix} [I\ \ I] = 
\begin{bmatrix}I\tensor I & I\tensor I \\ I\tensor I & I\tensor I\end{bmatrix}
\cong 
\begin{bmatrix}I & I \\ I & I\end{bmatrix}
$
\end{tabular}
\end{center}

We render the direct sum of matrices 
by layering.  Given objects $A, B \in \Cat_0$, 
\begin{center}
\begin{tabular}{c}
\raisebox{-0.4\height}{\ \includegraphics{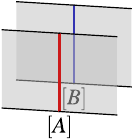}\ }
\\
$ [A] \boxplus [B] := \begin{bmatrix}A&O\\O&B\end{bmatrix} $
\end{tabular}
\end{center}

The 2-cells of $\MatC$ are rendered vertically in the upwards direction.
We horizontally compose 2-cells in the same way as 1-cells:
using the monoidal structure $\tensor,\oplus$ of $\Cat$.
For example, given morphisms $f:A\to B$ and $g:C\to D$ in 
$\Cat$, we can build $1\times 1$ matrix 2-cells $[f]:[A]\to[B]$
and $[g]:[C]\to[D]$, forming the horizontal composition and direct sum,
\begin{center}
$[f][g] =$
\raisebox{-0.4\height}{\ \includegraphics{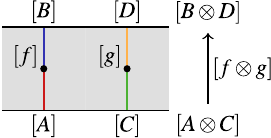}\ }
\hspace{20pt}
$[f]\boxplus[g] =$
\raisebox{-0.4\height}{\ \includegraphics{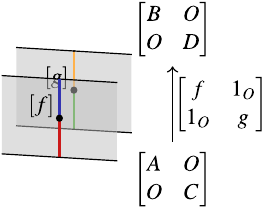}\ }
\end{center}
On the left we see that the grey surface supports the usual
string diagram calculus for the monoidal category
$(\Cat, \tensor, I, \alpha^\tensor, \lambda^\tensor, \rho^\tensor)$;
all we do is forget the distinction between a $1\times 1$
matrix 1-cell or 2-cell and its unique entry.

The $0\times 0$ matrix $[\ ]_{0,0}$ has no entries and corresponds
to the empty surface diagram.
This is a strict unit for the direct sum: 
$M\boxplus [\ ]_{0,0} = [\ ]_{0,0}\boxplus M = M$.
The $1\times 0$ matrix $[\ ]_{1,0}$ also has no entries, but 
participates in the direct sum by inserting
a zero row, for example $[A]\boxplus [\ ]_{1,0} = \begin{bmatrix}A\\O\end{bmatrix}$.
\begin{center}
\begin{tabular}{c l}
\raisebox{-0.4\height}{\ \includegraphics{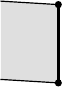}\ }
&
$
\hspace{35pt}
\raisebox{-0.4\height}{\ \includegraphics{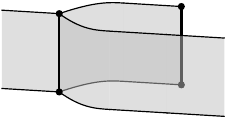}\ }
\cong
\raisebox{-0.4\height}{\ \includegraphics{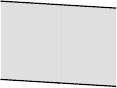}\ }
$
\\
\hspace{30pt}$[\ ]_{1,0}$
&
\hspace{35pt}
$[I\ \ I] \begin{bmatrix}I\\O\end{bmatrix} = [I\tensor I \oplus I\tensor O] \cong [I]$
\end{tabular}
\end{center}
We are being strict about the distinction between
equality on-the-nose ``$=$'' and isomorphism ``$\cong$''.
These isomorphisms are 2-cells, and are built by composing morphisms
in $\Cat$ that come from the bimonoidal structure of $\Cat$.
In fact, there will be an isomorphism 2-cell corresponding
to each of the equations in the presentation 
of the path calculus above \ref{path-calculus}.
The isomorphism 2-cells for {\tt(r-unit)} and {\tt(mul-hom)} are
rendered as
$$
\raisebox{-0.4\height}{\ \includegraphics{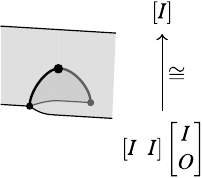}\ }
\hspace{30pt}
\raisebox{-0.4\height}{\ \includegraphics{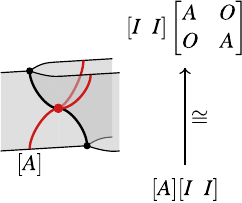}\ }
$$
respectively.
These are examples of \emph{coning}: 
forming a higher dimensional cone over two given string diagrams.
Being isomorphisms, these 2-cells obey equations.
For {\tt(r-unit)} we have
$$
\raisebox{-0.4\height}{\ \includegraphics{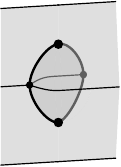}\ }
=
\raisebox{-0.4\height}{\ \includegraphics{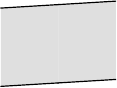}\ }
\hspace{30pt}
\raisebox{-0.4\height}{\ \includegraphics{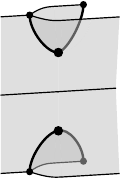}\ }
=
\raisebox{-0.4\height}{\ \includegraphics{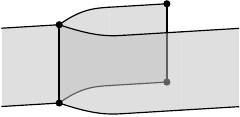}\ }
$$
The inverse isomorphism is rendered as the
vertically opposite diagram.
The isomorphism equations for {\tt(mul-hom)} 
we call
\emph{pulling} a string $[A]$ onto a pair of surfaces:
$$
\raisebox{-0.4\height}{\ \includegraphics{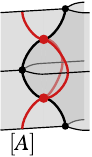}\ }
=
\raisebox{-0.4\height}{\ \includegraphics{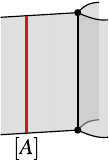}\ }
\hspace{20pt}
\raisebox{-0.4\height}{\ \includegraphics{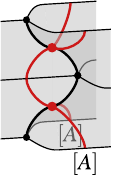}\ }
=
\raisebox{-0.4\height}{\ \includegraphics{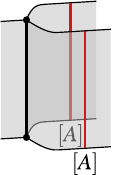}\ }
$$
And this isomorphism is 2-natural:
$$
\raisebox{-0.4\height}{\ \includegraphics{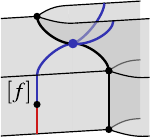}\ }
=
\raisebox{-0.4\height}{\ \includegraphics{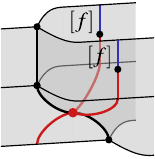}\ }
$$
Similar pulling equations hold for the horizontal opposite.

The horizontal composition of 1-cells given by
\begin{center}
\raisebox{-0.4\height}{\ \includegraphics{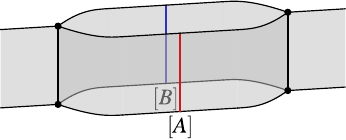}\ }
\\
$[I\ \ I] \begin{bmatrix}A&O\\O&B\end{bmatrix} \begin{bmatrix}I\\I\end{bmatrix}$
\end{center}
is isomorphic to the 1-cell $[A\oplus B]$.
Once again, we record this isomorphism as a 2-cell, which
is the categorified path calculus rule {\tt(add)}:
$$
\raisebox{-0.4\height}{\ \includegraphics{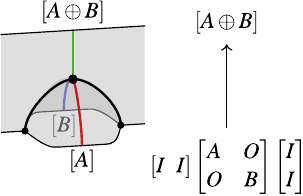}\ }
$$

Notice we have recovered the ``internal sum'' $\oplus$ of $\Cat$ by performing 
``external sum'' $\boxplus$ and horizontal compositions in $\MatC$.
This is the key idea behind this calculus: we have enlarged the domain of
definition to $\MatC$, a form of categorification, but then looking down from
above, we see certain expressions evaluate to $\Cat$ itself.
These are the microcosms discussed further in Section\ref{microcosms}.




\subsection{The symmetric monoidal structure}


The direct sum 2-functor $\boxplus:\MatC\times\MatC$ is a biproduct, and makes
$\MatC$ into a symmetric monoidal bicategory (section~\ref{symmetric-monoidal}).
The symmetry 
swaps surfaces (0-cells) around using permutation matrices.
We render this using a dotted line:
$$
\raisebox{-0.4\height}{\ \includegraphics{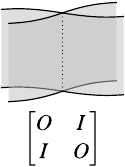}\ }
$$
This symmetry is natural on 1-cells and 2-cells. 
The naturality on 1-cells is exhibited by 2-cells:
$$
\raisebox{-0.4\height}{\ \includegraphics{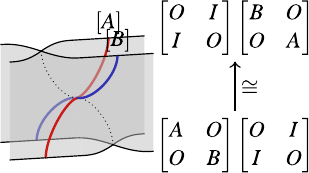}\ }
$$
and 2-naturality is the following
equation between 2-cells:
$$
\raisebox{-0.4\height}{\ \includegraphics{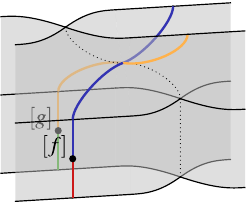}\ }
=
\raisebox{-0.4\height}{\ \includegraphics{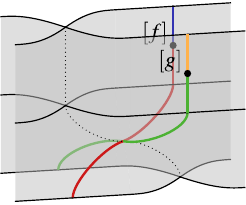}\ }
$$

Corresponding to the path calculus rule {\tt(comm)} 
we have isomorphism 2-cell
$$
\raisebox{-0.4\height}{\ \includegraphics{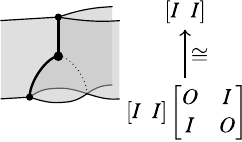}\ }
$$
and the horizontal opposite for {\tt(cocomm)}.
These are built from the unitors and nullitors of $\Cat$.


\subsection{The horizontal associator}\label{horizontal-associator}

So far we haven't
used the additive symmetry $\sigma^\oplus$, or
the left or right distributors $\delta^l$ and $\delta^r$ of $\Cat$.
These show up when we need associativity of the horizontal
composition (Definition~\ref{bicategory-def}).
Consider the following diagram,
$$
\raisebox{-0.4\height}{\ \includegraphics{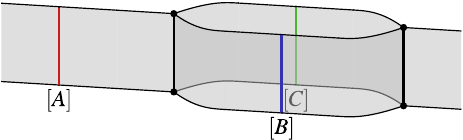}\ }
$$
We have many ways of horizontally composing this, and these
are related by \emph{horizontal associator} isomorphisms,
built using the bimonoidal structure of $\Cat$.
For example,
we use the left-distributor to exhibit the associator isomorphism:
\begin{align*}
[A]\Bigl([B\ \ C] \begin{bmatrix}I\\I\end{bmatrix}\Bigr) 
&\cong 
\Bigl( [A][B\ \  C] \Bigr)\begin{bmatrix}I\\I\end{bmatrix}, \\
[A\tensor(B\oplus C)] &\xrightarrow{[\delta^l_{A,B,C}]} [A\tensor B \oplus A\tensor C].
\end{align*}
To see where we need the additive symmetry, 
consider the horizontal composition of 1-cells:
$$
\raisebox{-0.4\height}{\ \includegraphics{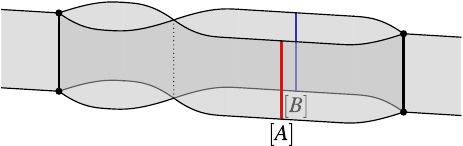}\ }
$$
$$
\begin{bmatrix}I&I\end{bmatrix}
\Biggl(
\begin{bmatrix}O&I\\I&O\end{bmatrix}
\begin{bmatrix}A&O\\O&B\end{bmatrix}
\Biggr)
\begin{bmatrix}I\\I\end{bmatrix}
\cong
\begin{bmatrix}I&I\end{bmatrix}
\begin{bmatrix}O&B\\A&O\end{bmatrix}
\begin{bmatrix}I\\I\end{bmatrix}.
$$
This expression now has two ways to be evaluated,
related by the additive symmetry of $\Cat$:
\begin{align*}
\Biggl(
\begin{bmatrix}I&I\end{bmatrix}
\begin{bmatrix}O&B\\A&O\end{bmatrix}
\Biggr)
\begin{bmatrix}I\\I\end{bmatrix}
\cong
\begin{bmatrix}A&B\end{bmatrix}
\begin{bmatrix}I\\I\end{bmatrix}
\cong
[A\oplus B],
\\
\begin{bmatrix}I&I\end{bmatrix}
\Biggl(
\begin{bmatrix}O&B\\A&O\end{bmatrix}
\begin{bmatrix}I\\I\end{bmatrix}
\Biggr)
\cong
\begin{bmatrix}I&I\end{bmatrix}
\begin{bmatrix}B\\A\end{bmatrix}
\cong
[B\oplus A].
\end{align*}
Note this isomorphism $[A\oplus B]\cong[B\oplus A]$
is an isomorphism 2-cell of $\MatC$.


\subsection{Additive biproducts in the base category}\label{biproducts}

When the additive structure in the base category $\Cat$ 
is a biproduct, we acquire a matrix calculus for morphisms of $\Cat$.
We use round brackets for this, such as the morphism
$
\begin{pmatrix}1\ \ 1\end{pmatrix} : I\oplus I \to I.
$
Promoting this to a $1\times 1$ matrix
in our surface diagram calculus yields the cap and cup:
$$
\raisebox{-0.4\height}{\ \includegraphics{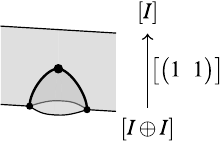}\ }
\hspace{35pt}
\raisebox{-0.4\height}{\ \includegraphics{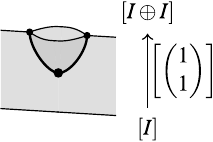}\ }
$$
Strictly speaking, we should be writing $[I\tensor I\oplus I\tensor I]$
instead of just $[I\oplus I]$. 
However, for the sake of clarity we
make this notational simplification.

Using this cap and cup we can define a \emph{trace}
which counts the number of layers, such as
$$
\mathrm{Tr}\Bigl( 
\raisebox{-0.4\height}{\ \includegraphics{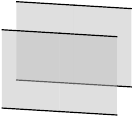}\ }
\Bigr) := 
\raisebox{-0.4\height}{\ \includegraphics{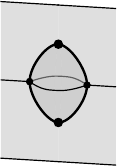}\ }
= 
\begin{bmatrix}\begin{pmatrix}1+1\end{pmatrix}\end{bmatrix}.
$$
We also have another cap and cup
$$
\raisebox{-0.4\height}{\ \includegraphics{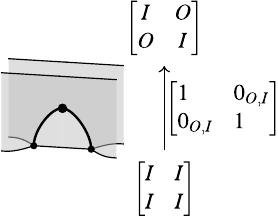}\ }
\hspace{35pt}
\raisebox{-0.4\height}{\ \includegraphics{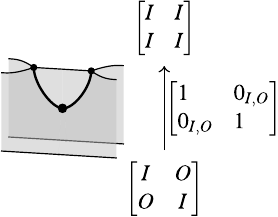}\ }
$$
that uses the zero maps $0_{O,I}:I\from O$ and $0_{I,O}:O\from I$ in $\Cat$.
All these caps and cups exhibit 
$$
\raisebox{-0.4\height}{\ \includegraphics{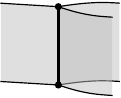}\ }
,\ \ \ 
\raisebox{-0.4\height}{\ \includegraphics{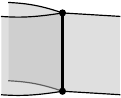}\ }
$$
as ambidextrous adjoints:
\begin{align*}
\raisebox{-0.4\height}{\ \includegraphics{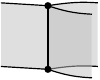}\ }
=
\raisebox{-0.4\height}{\ \includegraphics{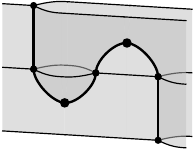}\ }
&,\ \ \ \ 
\raisebox{-0.4\height}{\ \includegraphics{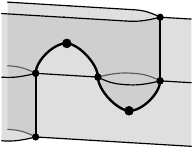}\ }
=
\raisebox{-0.4\height}{\ \includegraphics{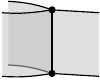}\ }
,\\
\raisebox{-0.4\height}{\ \includegraphics{mat_c.imcache/c13e979873cb69df12bf0a9c4081b0afe22480f0.pdf}\ }
=
\raisebox{-0.4\height}{\ \includegraphics{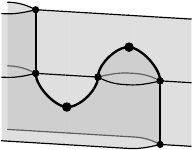}\ }
&,\ \ \ \ 
\raisebox{-0.4\height}{\ \includegraphics{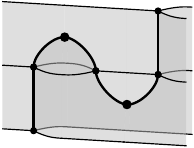}\ }
=
\raisebox{-0.4\height}{\ \includegraphics{mat_c.imcache/acd2b03855dbfdb7e430918de44b24b911e39114.pdf}\ }
.
\end{align*}
and we therefore have a \emph{Frobenius algebra} structure on $I\oplus I$ (See \cite{Lauda2006}).
It is straightforward to verify this is also also a special symmetric Frobenius
algebra, or a \emph{classical bit}~\cite{Vicary2012-1}.

We also have morphisms such as 
$
\begin{pmatrix}1_A\ \ 1_A\end{pmatrix} : A\oplus A \to A
$
in $\Cat$. Once again,
promoting this to a $1\times 1$ matrix
in our surface diagram calculus yields the 2-cells:
$$
\raisebox{-0.4\height}{\ \includegraphics{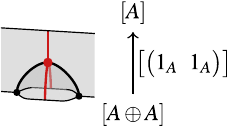}\ }
\hspace{35pt}
\raisebox{-0.4\height}{\ \includegraphics{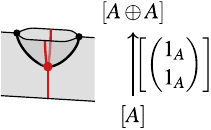}\ }
$$
Given morphisms $f,g : A\to B$, we can build the expression 
$$
[(f+g)] = 
\raisebox{-0.4\height}{\ \includegraphics{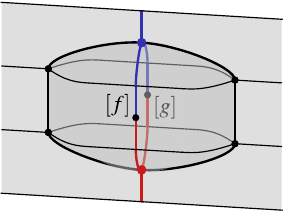}\ }
$$


\subsection{Duals in the base category}\label{duals}


When $\Cat$ has additive biproducts and duals (see~\cite{Heunen2019} \S 3.3) then $\MCat$
acquires all adjoints, given by the \emph{dual transpose} of matrices.
For example, given objects $A,B,C,D\in\Cat_0$
with duals $A^*,B^*,C^*,D^*\in\Cat_0$, exhibited by the counits and units
$$
\begin{array}{c c}
\epsilon_A:A\tensor A^*\to I, & \eta_A:I\to A^*\tensor A,\\
\epsilon_B:B\tensor B^*\to I, & \eta_B:I\to B^*\tensor B,\\
\epsilon_C:C\tensor C^*\to I, & \eta_C:I\to C^*\tensor C,\\
\epsilon_D:D\tensor D^*\to I, & \eta_D:I\to D^*\tensor D,\\
\end{array}
$$
we have the adjunction
$$
\begin{bmatrix}A&B\\C&D\end{bmatrix}
\dashv
\begin{bmatrix}A^*&C^*\\B^*&D^*\end{bmatrix}
$$
in $\MatC$, with counit given by
$$
\raisebox{-0.4\height}{\ \includegraphics{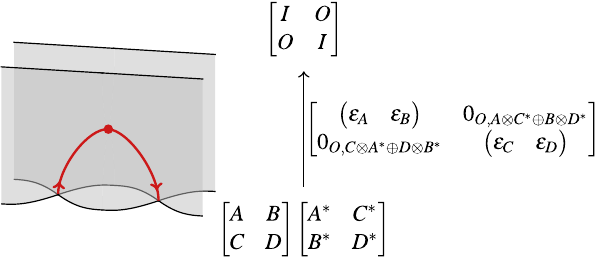}\ }
$$
and unit given by
$$
\raisebox{-0.4\height}{\ \includegraphics{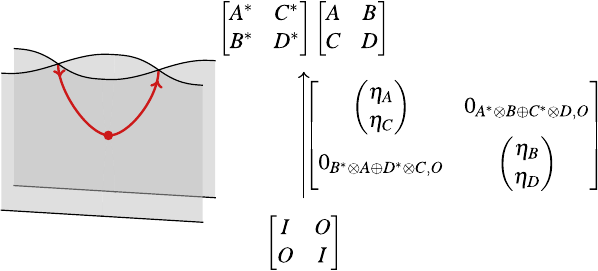}\ }
$$
The snake equations are rendered as
$$
\raisebox{-0.4\height}{\ \includegraphics{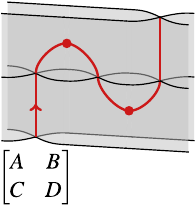}\ }
\lifteq
\raisebox{-0.4\height}{\ \includegraphics{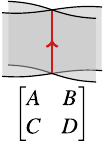}\ }
,\ \ \ \ \ \ \ \ \ \ 
\raisebox{-0.4\height}{\ \includegraphics{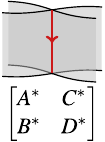}\ }
\lifteq
\raisebox{-0.4\height}{\ \includegraphics{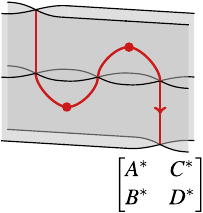}\ }
$$



\subsection{Traces}

When $\Cat$ has multiplicative symmetry, additive biproducts, and duals
then we can build traces in $\MCat$.
Recall that the trace of an object $A\in\Cat$
is the composite
$$
\mathrm{Tr}(A) :=
I \xrightarrow{\eta_A} A^*\tensor A 
\xrightarrow{\sigma^\tensor}
A\tensor A^* 
\xrightarrow{\epsilon_A} I
$$
which is rendered as
$$
\mathrm{Tr}\Biggl(
\raisebox{-0.4\height}{\ \includegraphics{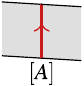}\ }
\Biggr) :=
\raisebox{-0.4\height}{\ \includegraphics{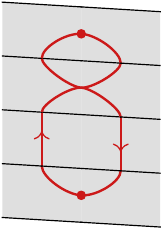}\ }
.$$


We could try to extend this diagram to arbitrary $m\times n$
matrices over $\Cat$, but we immediately run into a problem
with the symmetry:
matrix products do not commute in general. 
Already the products
$$
[A\ \ B] \begin{bmatrix}C\\D\end{bmatrix} = [A\tensor C\oplus B\tensor D],\ \ 
\begin{bmatrix}C\\D\end{bmatrix} [A\ \ B] = 
\begin{bmatrix}C\tensor A & C\tensor B \\ D\tensor A & D\tensor B \end{bmatrix}
$$
live on a different number of layers so we cannot
hope to swap them using a 2-cell isomorphism.~\footnote{
Another solution is to use a different graphical calculus~\cite{Ponto2013}}
The trick here is that 
something like this does work ``under the cap''.
To demonstrate this we prove the identity
$\Tr(A\oplus B) = \Tr(A)+\Tr(B):$

\begin{align*}
\Tr\Biggl(
\raisebox{-0.4\height}{\ \includegraphics{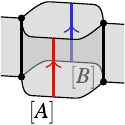}\ }
\Biggr)
&\ \ \ =
\raisebox{-0.4\height}{\ \includegraphics{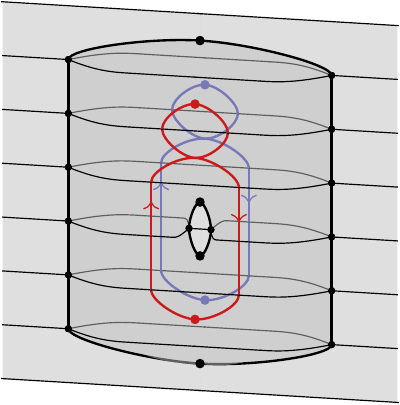}\ }
=
\raisebox{-0.4\height}{\ \includegraphics{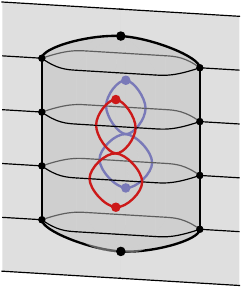}\ }
\\
&\ \ \ =\Tr\Biggl(
\raisebox{-0.4\height}{\ \includegraphics{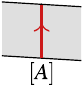}\ }
\Biggr) + \Tr\Biggl(
\raisebox{-0.4\height}{\ \includegraphics{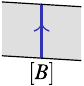}\ }
\Biggr).
\end{align*}


\subsection{2-categorical quantum mechanics}\label{teleportation}


In this section we let $\Cat$ be compact closed with biproducts
and a compatible dagger structure~\cite{Abramsky2004}.
In this case 
$\MatC$ inherits a dagger structure (section~\ref{dagger}).
This dagger is the identity on 0-cells and 1-cells, 
contravariant on 2-cells, and squares to the identity.
For example,
given $f:A\to B$ and $g:B\to D$ in $\Cat$:
$$
\Biggl(
\raisebox{-0.4\height}{\ \includegraphics{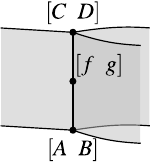}\ }
\Biggr)^\dag
:=
\raisebox{-0.4\height}{\ \includegraphics{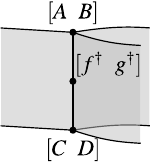}\ }
.$$

A \emph{unitary} 2-cell $\mu$ in $\MatC$
is an isomorphism whose inverse is $\mu^\dag$.
Following \cite{Vicary2012-1}, 
we define the \emph{qubit measurement, qubit preparation} and 
\emph{projective measurement}
on $A\in\Cat_0$, as unitary 2-cells
of the form:
\begin{center}
\begin{tabular}{ c c c }
\ \ \ \ \raisebox{-0.4\height}{\ \includegraphics{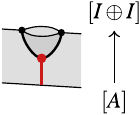}\ }
&
\ \ \ \ \raisebox{-0.4\height}{\ \includegraphics{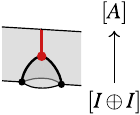}\ }
&
\ \ \ \ \raisebox{-0.4\height}{\ \includegraphics{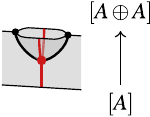}\ }
\\
qubit measurement & qubit preparation & projective measurement \\
\end{tabular}
\end{center}
A similar definition holds for $n$-dimensional systems, here we have
rendered the case $n=2$. 
It is a straightforward exercise to
show graphically that if a system $A$ has an $n$-dimensional
measurement then $A$ itself is $n$-dimensional in the sense of $\Tr(A)=n$.

The pulling equations are fundamental to 2-categorical quantum
mechanics. We see graphically how a 
quantum system bifurcates onto classical surfaces:
$$
\raisebox{-0.4\height}{\ \includegraphics{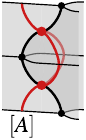}\ }
=
\raisebox{-0.4\height}{\ \includegraphics{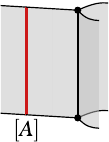}\ }
\hspace{20pt}
\raisebox{-0.4\height}{\ \includegraphics{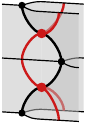}\ }
=
\raisebox{-0.4\height}{\ \includegraphics{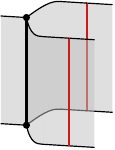}\ }
\hspace{20pt}
\raisebox{-0.4\height}{\ \includegraphics{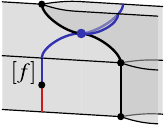}\ }
=
\raisebox{-0.4\height}{\ \includegraphics{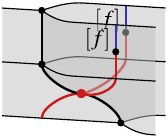}\ }
$$
We use these classical surfaces to define controlled operations
as follows.
Given unitaries $U_1,U_2:A\to A$ in $\Cat$ the \emph{controlled}-$U$ 
operation is the 2-cell given by
$$
\raisebox{-0.4\height}{\ \includegraphics{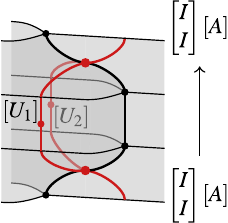}\ }
$$
A straightforward calculation using the pulling equations
shows that this is unitary.
A similar definition holds for higher arity control operations:
an $n$-ary control operation will have the $A$ system pulled
onto $n$ surfaces where the unitaries $U_1,...,U_n$ can act.

Note that these unitary 2-cells are all primitive 2-cells in the graphical 
calculus of \cite{Vicary2012-1};
here we are seeing more of the internal structure involved in these constructions.

Recall that a \emph{teleportation protocol}
on a system $A\in\Cat_0$
is a measurement $\mu$ on $A\tensor A^*$ and a control
operation $\gamma$ on $A$, such that:
$$
\raisebox{-0.4\height}{\ \includegraphics{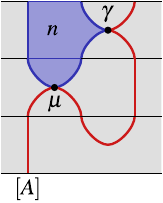}\ }
\ \raisebox{15pt}{$=\frac{1}{\sqrt{n}}$}\ 
\raisebox{-0.4\height}{\ \includegraphics{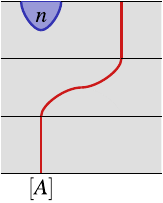}\ }
$$
Here we have labelled the blue region $n$ to denote $n$ copies of 
the grey surface: $n = \boxplus_{1}^n 1,$
and $\frac{1}{\sqrt{n}}$ denotes a morphism $I\to I$ in $\Cat$.


A \emph{unitary error basis} for a system $A\in\Cat_0$
is a sequence of unitaries $U_1:A\to A,...,U_n:A\to A$ in $\Cat_1$
such that the 2-cell 
$$
\mu =
\frac{1}{\sqrt{n}}
\raisebox{-0.4\height}{\ \includegraphics{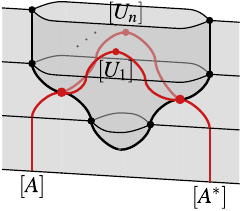}\ }
$$
is unitary.
(For clarity, we only show two surfaces, and also we ``slide''
$U_i$ onto the counit as $U_i:A\tensor A^*\to I$.)
This means that $\mu$ is 
a measurement on $A\tensor A^*$. 
If we further define $\gamma$ as the controlled-$(U_1^\dag,...,U_n^\dag)$ operation,
then we have the following teleportation protocol:
\begin{align*}
\raisebox{-0.4\height}{\ \includegraphics{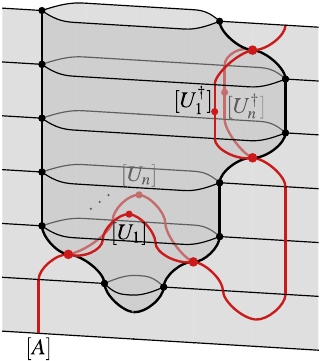}\ }
=
\raisebox{-0.4\height}{\ \includegraphics{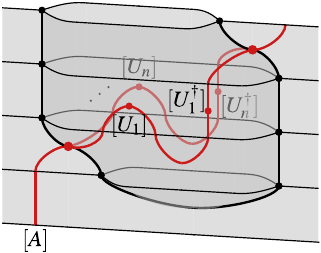}\ }
&=
\raisebox{-0.4\height}{\ \includegraphics{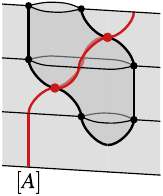}\ }
\\
&=
\raisebox{-0.4\height}{\ \includegraphics{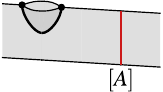}\ }
\end{align*}

Our graphical calculus is fine-grained enough that
we can use it to construct unitaries themselves.
For example, we define the \emph{qubit Pauli X} operation as the 2-cell
$$
\raisebox{-0.4\height}{\ \includegraphics{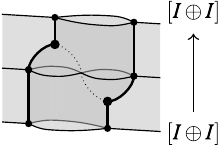}\ }
\ \ = \ \ 
\raisebox{-0.4\height}{\ \includegraphics{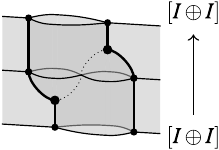}\ }
.$$
This definition contains a crucial use 
of the horizontal associator as discussed in \S\ref{horizontal-associator}.


\subsection{Relation to other graphical calculi}\label{relation}

We examine the sheet diagrams from ~\cite{Comfort2020}.
Given objects $A, B, C, D$ and morphisms
$
f : A\oplus A\tensor B \to C,\ \ 
1_C : C\to C,\ \ 
g : B\tensor D \to A\oplus D
$
from a bimonoidal category $\Cat$,
the morphism  $f\tensor 1_C\tensor g:$ is rendered:
\begin{center}
\includegraphics[]{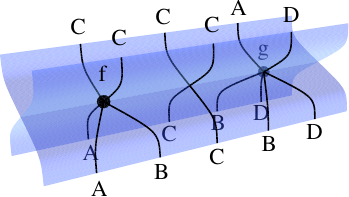}
\end{center}

Using our surface diagram calculus the same morphism
$f\tensor 1_C\tensor g$ is rendered as,
$$
\raisebox{-0.4\height}{\ \includegraphics{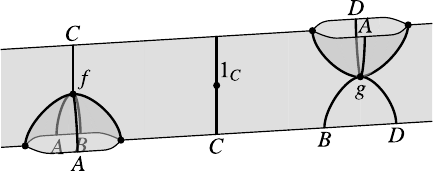}\ }
$$
All the strings here are labelled with $1\times 1$ matrices
so we neglect the bracket notation, such as $[A]$, etc.
Up to reversing the front-to-back ordering of layers,
the sheet diagram is seen to correspond to a normal form for
the more general surface diagrams. The general procedure
is to expand the bubbles, or pull strings onto the bubbles.
For example, we pull the string for the
domain of $1_C$ onto the domain of $f$:
$$
\raisebox{-0.4\height}{\ \includegraphics{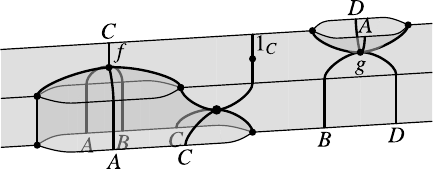}\ }
$$
and similarly for the $B$ and $D$ strings in the domain of $g$. 
Then we pull the codomain of $f$ and $1_C$ onto the codomain of $g$.

Another graphical calculus for bimonoidal categories
is given by tape diagrams \cite{Bonchi2022}.

Also, there is the $ZXW$ calculus which is a two dimensional string diagram 
calculus that combines a path calculus (for direct sums) and the
monoidal tensor product structure of vector spaces \cite{deFelice2023}.


\subsection{Where are we?}

%

String diagrams for monoidal categories
are well understood~\cite{Joyal1991,Selinger2010}, and these readily 
generalize (via folklore~\cite{Ponto2013}) to string
diagrams for bicategories.

Our graphical calculus fits into the
framework of the surface diagram calculus for symmetric monoidal
bicategories, see~\cite{Gurski2013} and discussion in ~\cite{Campbell2022} \S 4.1.

What we have is an \emph{additive monoidal bicategory}
$\Mat(\Cat)^\boxplus$, where the monoidal product $\boxplus$
is the direct sum of matrices.
There is also a \emph{multiplicative monoidal bicategory}
$\Mat(\Cat)^\boxtimes$ which uses a categorified Kronecker product $\boxtimes$ of matrices,
see \cite{Johnson2021} \S 8.12. 
The surface diagram calculus for $\Mat(\Cat)^\boxtimes$ is used in the context of
2-Hilbert spaces, see for example ~\cite{Heunen2019} \S 8.2, and\cite{Reutter2019biunitary}.
Putting these together we expect to have a 
\emph{bimonoidal bicategory} $\MatC^{\boxplus,\boxtimes}$.

We began this work with the question: what is an amplitude?
This kind of \emph{what is} question 
-- what is a group, what is a space, etc. -- 
is answered by functorially representing into a semantic category. 
For us, the semantic category is a category of modules;
we have been secretly dealing with the free bimodule
bicategory over the category $\Cat$, which directly categorifies
linear algebra over the field of complex numbers (amplitudes).
By taking microcosms we recover $\Cat$ as the Hom category $\Hom(1,1)$
of the free bimodule bicategory $\MatC$.

\section{The bicategory of matrices over a bimonoidal category}\label{bicategory}

In this section we start over from the beginning, with a more
rigorous abstract exposition. The notation is somewhat different,
although consistent with the above. 

\begin{warning}
To be consistent with matrix notation, and
the horizontal direction of the graphical calculus,
the Hom categories are defined right-to-left:
$\Hom(m,n)$ consists of 1-cells $m\leftarrow n$.
\end{warning}

\begin{definition}~\label{matc-bicategory}
(\cite{Johnson2021} \S 8.1)
Given a category $\Cat$
and natural numbers $m,n=0,1,2...,$ we define the
\emph{matrix category} $\Mat_{m,n}(\Cat)$:

\Item
Objects are $m\times n$ matrices $A = [A_{ij}]$
with entries $A_{ij}\in\Cat_0$.
We render such an object as
$$
\raisebox{-0.4\height}{\ \includegraphics{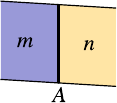}\ }
$$

\Item
Morphisms $f:A\to B$
are matrices $[f_{ij}:A_{ij}\to B_{ij}]$ of
morphisms of $\Cat$.
These are rendered as
$$
\raisebox{-0.4\height}{\ \includegraphics{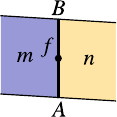}\ }
$$

\Item
Composition of morphisms $f:A\to B$ and $g:B\to C$ is
defined by componentwise composition in $\Cat$,
$(g\circ f)_{ij} := g_{ij}\circ f_{ij}.$
This composition is rendered vertically as
$$
\raisebox{-0.4\height}{\ \includegraphics{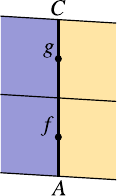}\ }
$$

\Item
The identity morphisms, $1_A:A\to A$ are defined as
the componentwise identities from $\Cat$,
$(1_A)_{ij} := 1_{A_{ij}}$.

\tombstone
\end{definition}

\begin{definition}\label{bicategory-def}
(\cite{Johnson2021} \S 8)
Given a bimonoidal category  $\Cat$,
we define the bicategory $(\Mat(\Cat),\alpha,\lambda,\rho)$
as follows:

\Item The objects, or 0-cells, are natural numbers $n=0,1,2...$
These are rendered as shaded regions:
$$
\raisebox{-0.4\height}{\ \includegraphics{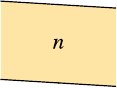}\ }
$$

\Item
For objects $m$ and $n$ of $\Mat(\Cat)$,
the hom categories are
defined as $\Hom(m,n) := \Mat_{m,n}(\Cat)$.

\Item
Horizontal composition is a functor
$$
    \Hom(l,m) \times \Hom(m,n) \to \Hom(l,n),
$$
given by \emph{matrix multiplication}.
For $A\in\Hom(l,m), B\in\Hom(m,n)$,
this is defined using the bimonoidal structure of $\Cat$:
$$
    (AB)_{ik} := \bigoplus_j A_{ij}\tensor B_{jk}.
$$
We render this horizontal composite as
$$
\raisebox{-0.4\height}{\ \includegraphics{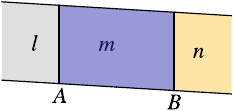}\ }
$$

\Item
Given object $m$, the unit 1-cell $1_m \in \Hom(m,m)$ is 
$$
    (1_m)_{ij} =
\begin{cases}
    1^\otimes \ \ \mathrm{if} \ i=j,  \\
    0^\oplus \ \ \mathrm{otherwise}.
\end{cases}
$$
This is rendered as:
$$
\raisebox{-0.4\height}{\ \includegraphics{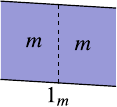}\ }
$$

\Item For objects $m,n$, the left unitor $\lambda_{m,n}$ has component
at 1-cells $A\in\Hom(m,n)$ the 2-cell $\lambda_{m,n}(A): 1_m A\to A$.
This is built by composing the left nullitors, 
additive unitors and left multiplicative unitor of $\Cat$.
Similarly for the right unitor $\rho_{m,n}$, which is
built by composing the right nullitors, 
additive unitors and right multiplicative unitors of $\Cat$.

\Item For objects $l,m,n,o$, the associator $\alpha_{l,m,n,o}$ 
has components at $A\in\Hom(l,m)$, $B\in\Hom(m,n)$, $C\in\Hom(n,o)$
given by the 2-cell $\alpha_{l,m,n,o}(A,B,C):(AB)C\to A(BC)$.
The $(i,p)$ entry of this matrix is a morphism in $\Cat$:
$$
     \bigoplus_k \bigl(\bigoplus_j A_{ij}\tensor B_{jk}\bigr)\tensor C_{kp}
    \to
     \bigoplus_j A_{ij}\tensor \bigl(\bigoplus_k B_{jk}\tensor C_{kp}\bigr)
$$
built by composing the left and right distributors,
additive associativity, additive symmetry, and multiplicative associators of $\Cat$.
\tombstone
\end{definition}

Given a bimonoidal category $\Cat$, 
we define the matrix $0_{m,n}\in\Mat_{m,n}(\Cat)$
as the matrix $[A_{ij}]$ with $A_{ij}=O$,
and we define the matrix $1_{m}\in\Mat_{m,m}(\Cat)$
as the matrix $[A_{ij}]$ with $A_{ii}=I$, and $A_{ij}=O$ for $i\ne j$.

\subsection{Biproducts in $\MatC$}

\begin{definition}
We define a 2-functor
$$
    \boxplus: \Mat(\Cat)\times\Mat(\Cat) \to \Mat(\Cat)
$$
as follows:

\Item
On 0-cells $m,n$ we have $m \boxplus n := m+n.$

\Item
On 1-cells $A:m\leftarrow n$ and $B:m'\leftarrow n'$,
the block-diagonal matrix 
$$
    A\boxplus B := \begin{bmatrix} A & 0 \\ 0 & B \end{bmatrix}.
$$

\Item
On 2-cells $f:A\to B$ and $g:C\to D$,
the block-diagonal matrix 
$$
    f\boxplus g := \begin{bmatrix} f & 0 \\ 0 & g \end{bmatrix}.
$$

\tombstone
\end{definition}

\begin{theorem}\label{matc-biproducts}
Given a bimonoidal category $\Cat$,
the bicategory $\Mat(\Cat)$ has all finite biproducts
and a zero object.
\end{theorem}
\begin{proof}

\Item
We claim that $0$ is a zero object.
To see that $0$ is an initial object,
we find the category $\Hom(n,0)=\Mat_{n,0}(\Cat)$ is
terminal on-the-nose. It has one
object which is the unique $n\times 0$ matrix $[\ ]_{n,0}$
and one identity morphism, also an empty matrix:
$$
\raisebox{-0.4\height}{\ \includegraphics{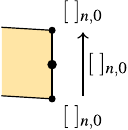}\ }
$$
The horizontally opposite argument shows that $0$ is also a terminal object.

\Item
Given objects $m,n$ in $\MatC$, we show that
$m\boxplus n$ is the product of $m$ and $n$.
This means we must show a natural equivalence of
categories $\Hom(m\boxplus n, l) \cong \Hom(m,l)\times\Hom(n,l)$
for all objects $l$.
But this follows immediately from the definition of the $\Hom$ categories
by concatenation of matrices,
$\Mat_{m+n,l}(\Cat) \cong \Mat_{m,l}(\Cat)\times \Mat_{n,l}(\Cat)$.
The horizontally opposite argument shows that $m\boxplus n$
is the coproduct of $m$ and $n$, and so $m\boxplus n$ is a biproduct.
By iterating this we obtain all finite biproducts.
\end{proof}

\subsection{Symmetric monoidal structure in $\MatC$}\label{matc-symmetric}

\newcommand{\Zero}{0}

\begin{theorem}\label{symmetric-monoidal}
Given a bimonoidal category $\Cat$, then
$\MatC^\boxplus := (\Mat(\Cat), \boxplus, \Zero)$
is a symmetric monoidal bicategory.
\end{theorem}
\begin{proof}
By the previous Theorem ~\ref{matc-biproducts},
$\MatC$ is a bicategory with binary product $\boxplus$
and terminal object $0$. The result follows
by \cite{Carboni2008} Theorem 2.15.
\end{proof}

The monoidal product is strict. On 0-cells $l,m,n$ of $\Mat(\Cat)$
\begin{align*}
    (l\boxplus m)\boxplus n &= l\boxplus (m\boxplus n), \\
    0 \boxplus l &= l, \\
    l\boxplus 0  &= l.
\end{align*}
And similarly for the 1- and 2-cells of $\Mat(\Cat)$.
The monoidal product of 0,1,2-cells is rendered by layering cells:
\begin{align*}
\raisebox{-0.4\height}{\ \includegraphics{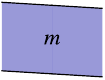}\ }
\boxplus
\raisebox{-0.4\height}{\ \includegraphics{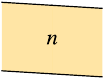}\ }
&=
\raisebox{-0.4\height}{\ \includegraphics{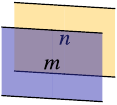}\ }
\\
\raisebox{-0.4\height}{\ \includegraphics{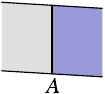}\ }
\boxplus
\raisebox{-0.4\height}{\ \includegraphics{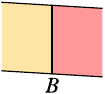}\ }
&=
\raisebox{-0.4\height}{\ \includegraphics{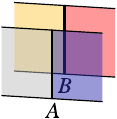}\ }
\\
\raisebox{-0.4\height}{\ \includegraphics{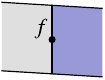}\ }
\boxplus
\raisebox{-0.4\height}{\ \includegraphics{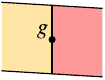}\ }
&=
\raisebox{-0.4\height}{\ \includegraphics{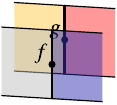}\ }
\end{align*}

The components of the \emph{symmetry}
2-natural isomorphism, $\sigma_{m,n}^\boxplus: n\boxplus m\leftarrow m\boxplus n $
are given by the $(n+m)\times (m+n)$ block matrix:
$$
    \sigma_{m,n}^\boxplus := \begin{bmatrix} 0_{n,m} & 1_n \\ 1_m & 0_{m,n} \end{bmatrix}.
$$
and rendered as the 1-cell:
$$
\raisebox{-0.4\height}{\ \includegraphics{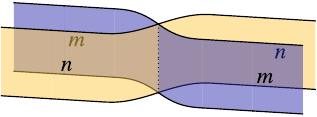}\ }
$$

The \emph{syllepsis} 2-cell
$
\nu_{m,n}^\boxplus:
\sigma_{n\boxplus m}^\boxplus \sigma_{m\boxplus n}^\boxplus \to 1_{m\boxplus n}
$
is defined by composing unitors and nullitors of $\Cat$.
We render this as
$$
\raisebox{-0.4\height}{\ \includegraphics{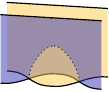}\ }
.
$$

%
%
%

\subsection{The category of scalars}\label{microcosms}

An important idea here is given by \emph{taking microcosms}:
the 1-categorical level is found \emph{in toto} within the 
2-categorical structure, rather than as a truncation (quotient) thereof.
These microcosms will be Homs of the 2-category.
The terminology is in relation to the \emph{microcosm principle}
\cite{Baez1998}.

\begin{theorem}
Given a bimonoidal category $\Cat$, the
hom category $\Hom(1,1)$ in $\Mat(\Cat)$ has a bimonoidal structure 
such that $\Hom(1,1)$ is bimonoidally equivalent to $\Cat$.
\end{theorem}
\begin{proof} We outline this as follows:

\Item 
The category 
$\Hom(1,1) = \Mat_{1,1}(\Cat)$ 
is isomorphic to $\Cat$ after removing the matrix brackets.

\Item
The horizontal composition 
$\Hom(1,1)\times\Hom(1,1)\to\Hom(1,1)$
gives 
$\Hom(1,1)$ a monoidal structure, which is
identical to the monoidal structure 
$(\otimes, I, \alpha^\otimes, \lambda^\otimes, \rho^\otimes)$ of $\Cat$.

\Item
We define the additive monoidal structure
$\oplus:\Hom(1,1)\times\Hom(1,1)\to\Hom(1,1)$ 
using the 1-cells given by the diagonal
$\begin{bmatrix}I\\I\end{bmatrix}$
and codiagonal $\begin{bmatrix}I&I\end{bmatrix}$.

\Item
Similarly, the symmetry $\sigma^\boxplus$ gives a
symmetry $\sigma^\oplus$ by using the diagonal and codiagonal 1-cells.
\end{proof}



\subsection{The dagger}\label{dagger}

\begin{definition}
(\cite{Karvonen2019})
A dagger category $(\Cat, \dag)$ is a 
category $\Cat$ together with a 
functor $\dag:\Cat\to\Cat^{\mathrm{op}}$
that is involutive and identity  on objects:
$f^{\dag\dag}=f$ and $A^\dag=A$ for morphisms $f\in\Cat_1$
and objects $A\in\Cat_0$.
\tombstone
\end{definition}

\begin{definition}
In a dagger category $(\Cat,\dag)$, a morphism $f\in\Cat_1$ 
is \emph{unitary} when $f^{-1} = f^\dag$.
\tombstone
\end{definition}

\begin{definition}
A \emph{dagger monoidal} category $\Cat$ is
a monoidal category 
$ (\Cat, \tensor, I, \alpha^\tensor, \lambda^\tensor, \rho^\tensor) $
and a dagger category $(\Cat,\dag)$,
such that $(f\tensor g)^\dag = f^\dag\tensor g^\dag$
for all morphisms $f,g\in\Cat_1$, and 
the natural isomorphisms 
$\alpha^\tensor, \lambda^\tensor, \rho^\tensor$
have unitary components.
Similarly, a \emph{dagger symmetric monoidal} category,
has symmetry $\sigma^\tensor$ with unitary components.
\tombstone
\end{definition}

\begin{definition}
A \emph{dagger bimonoidal} category $\Cat$ is
a bimonoidal category $\Cat$ whose additive symmetric monoidal structure
is dagger symmetric monoidal 
and whose multiplicative monoidal structure is dagger monoidal.
\end{definition}

\begin{definition}
A \emph{dagger bicategory} $(\Bicat,\alpha,\lambda,\rho,\dag)$
is a bicategory $(\Bicat,\alpha,\lambda,\rho)$ such that:

\Item
The categories $\Hom(m,n)$ are dagger categories:
for objects $m,n\in\Bicat$ we have a 
functor 
$$\dag_{m,n}:\Hom(m,n)\to\Hom(m,n)^{\mathrm{op}}
$$
that is the identity on objects
(1-cells of $\Bicat$) and 
$\dag_{m,n}^{\mathrm{op}}\dag_{m,n}$ is the identity functor.

\Item The natural isomorphisms $\alpha,\lambda,\rho$ have
unitary components:
$$
\alpha_{l,m,n,o}^{-1} = (\alpha_{l,m,n,o})^{\dag_{l,o}},\ \ 
\lambda_{l,m}^{-1} = (\lambda_{l,m})^{\dag_{l,m}},\ \ 
\rho_{l,m}^{-1} = (\rho_{l,m})^{\dag_{l,m}},
$$
for objects $l,m,n,o$ of $\Bicat$.
\tombstone
\end{definition}

\begin{theorem}\label{dagger-bicategory}
Given a dagger bimonoidal category $\Cat$, 
$(\Mat(\Cat), \alpha,\lambda,\rho,\dag)$
is a dagger bicategory, where we define the functors
$$
\dag_{m,n}:\Mat_{m,n}(\Cat) \to \Mat_{m,n}(\Cat)^{\mathrm{op}}
$$
by applying $\dag:\Cat\to\Cat^{\mathrm{op}}$ componentwise.
\end{theorem}
\begin{proof}
Componentwise application of $\dag:\Cat\to\Cat^{\mathrm{op}}$ 
is easily seen to give a dagger structure $\dag_{m,n}$ on
the categories $\Hom(m,n)$.
From Definition~\ref{matc-bicategory},
we see the associator and left and right unitor are built by
composing the bimonoidal structure morphisms of $\Cat$.
These are all unitary and so the
associator and left and right unitor of the bicategory $\Mat(C)$ 
will have unitary components.
\end{proof}

\begin{definition}
A \emph{dagger symmetric monoidal bicategory} 
$(\Bicat, \boxplus, \Zero, \sigma^\boxplus, \nu^\boxplus, \dag)$
is a dagger bicategory $(\Bicat,\alpha,\lambda,\rho,\dag)$
and a symmetric monoidal bicategory 
$(\Bicat, \boxplus,\Zero, \sigma^\boxplus, \nu^\boxplus)$
such that
$(f\boxplus g)^\dag = f^\dag\boxplus g^\dag$, for 2-cells $f,g\in\Bicat_2$
and the components of the syllepsis $\nu^\boxplus$ are unitary.
\end{definition}

\begin{theorem}\label{symmetric-monoidal-dagger}
Given a dagger bimonoidal category $\Cat$, 
$(\Mat(\Cat), \boxplus, \Zero, \sigma^\boxplus, \nu^\boxplus, \dag)$
is a dagger symmetric monoidal bicategory.
\end{theorem}
\begin{proof}
The dagger $\dag_{m,n}$ acts componentwise on $\Mat_{m,n}(\Cat)$
and this commutes with $\boxplus$.
The components of $\nu^\boxplus$ are compositions of unitary morphisms
of $\Cat$.
\end{proof}

{\bf Acknowledgements.}
I do wish to acknowledge some very useful discussions
with colleagues and friends; I believe it is appropriate
to do this in final versions of papers.

\bibliography{refs}{}
\bibliographystyle{alpha}

\end{document}